\RequirePackage{fix-cm}
\RequirePackage{amsmath}
\documentclass[sn-mathphys,Numbered]{sn-jnl}

\usepackage[utf8]{inputenc} 
\usepackage[T1]{fontenc}
\usepackage{url}
\usepackage{ifthen}
\usepackage{amssymb}
\usepackage{amsthm}
\usepackage{amsfonts}
\usepackage{marvosym}
\usepackage{mathrsfs}
\usepackage[title]{appendix}
\usepackage{graphicx}
\usepackage{indentfirst}
\usepackage{bm}
\usepackage{verbatim}
\usepackage{yfonts}
\usepackage{bbm}
\usepackage{multirow}
\usepackage{float,tabularx}
\usepackage{enumitem}
\usepackage{changepage}
\usepackage{adjustbox}
\usepackage{makecell}
\usepackage{framed}
\usepackage{varwidth}
\usepackage{xcolor}
\usepackage{hyperref}
\usepackage{ifsym}
\usepackage{etex}
\usepackage{textcomp}%
\usepackage{manyfoot}%
\usepackage{booktabs}%
\usepackage{algorithm}%
\usepackage{algorithmicx}%
\usepackage{algpseudocode}%
\usepackage{listings}%

\theoremstyle{definition}%
\newtheorem{theorem}{Theorem}
\newtheorem{prop}[theorem]{Proposition}%

\newtheorem{lem}[theorem]{Lemma}
\newtheorem{cor}[theorem]{Corollary}

\theoremstyle{definition}%
\newtheorem{rem}{Remark}%
\newtheorem{dfn}{Definition}%

\newcommand{\rk}{\text{rk}}
\newcommand{\colrk}{\text{colrk}}
\newcommand{\Gab}{\text{Gab}}

\newcommand{\F}{\mathbb{F}}

\newcommand{\circu}[2]{\text{Cir}_{#1}({#2})}

\newcommand{\Gpub}{G_{\sf pub}}

\newcommand{\ra}{\overset{\$}{\leftarrow}}

\newcommand{\size}{{\sf size}}
\newcommand{\Setup}{{\sf Setup}}
\newcommand{\Keygen}{{\sf KeyGen}}

\newcommand{\pk}{{\sf pk}}
\newcommand{\sk}{{\sf sk}}

\newcommand{\Encrypt}{{\sf Encrypt}}
\newcommand{\Decrypt}{{\sf Decrypt}}

\begin{document}

\title{On the Security and Design of Cryptosystems Using Gabidulin-Kronecker Product Codes} 

\author[1]{\fnm{Terry Shue Chien} \sur{LAU}}\email{terry.lau@mmu.edu.my}

\author[2]{\fnm{Zhe} \sur{SUN}}\email{zhe\_sun@mail.nankai.edu.cn}

\author*[3]{\fnm{Sook-Chin} \sur{YIP}} \email{scyip@mmu.edu.my}	

\author[4]{\fnm{Ji-Jian} \sur{CHIN}} \email{ji-jian.chin@plymouth.ac.uk}

\author[1]{\fnm{Choo-Yee} \sur{TING}} \email{cyting@mmu.edu.my}

\affil[1]{\orgdiv{Faculty of Computing and Informatics}, \orgname{Multimedia University}, \orgaddress{\city{Cyberjaya}, \postcode{63100}, \state{Selangor}, \country{Malaysia}}}

\affil[2]{\orgdiv{Chern Institute of Mathematics and LPMC}, \orgname{Nankai University}, \orgaddress{\city{Tianjin}, \postcode{300071}, \country{China}}}

\affil[3]{\orgdiv{Faculty of Engineering},  \orgname{Multimedia University}, \orgaddress{\city{Cyberjaya}, \postcode{63100}, \state{Selangor}, \country{Malaysia}}}

\affil[4]{\orgdiv{School of Engineering, Computing and Mathematics (Faculty of Science and Engineering)},  \orgname{University of Plymouth}, \orgaddress{\city{Drake Circus, Plymouth}, \postcode{PL 48AA}, \country{United Kingdom}}}


\abstract{This paper is a preliminary study on the security and design of cryptosystems using Gabidulin-Kronecker Product Codes. In particular, we point out the design impracticality of the system, and propose ways to improve it.}


\maketitle


\maketitle

\section{GabKron Cryptosystem}
 In this section, we describe the specifications for the McEliece-type cryptosystem based on Gabidulin-Kronecker product codes, namely the GabKron cryptosystem.
 
\noindent $\Setup(\delta)$. For a given security level $\delta$, choose $q$, a power of prime, and positive integers $m$, $n$, $k$, $n_{1}$, $n_{2}$, $k_{1}$, $k_{2}$, $t$, $t_{2}$, $\lambda$ satisfying $k < n \leq m$, $n= n_{1}n_{2}$, $k = k_{1}k_{2}$, $t \leq n_{1}t_{2}, 1 \leq k_{1} \leq n_{1},1 \leq k_{2} < n_{2}$, $2 \leq \lambda$. Let $\mathcal{V}$ be a $\lambda$-dimensional subspace of $\F_{q^m}$. Output parameters $param = (q,m,n,k,t,\lambda,n_{1},n_{2},k_{1},k_{2},\lambda)$.

\noindent $\Keygen(q,m,n,k,t,\lambda,n_{1},n_{2},k_{1},k_{2},\lambda)$. Randomly generate the following:
\begin{enumerate}
\item A normal element $\alpha \ra \mathbb{F}_{q^m}$. Let $\bm{g_{1}} = (\alpha^{[n_{1}-1]},\alpha^{[n_{1}-2]},\ldots,\alpha)$. Form a generator matrix $G_{1} = Cir_{k_{1}}(\bm{g_{1}}) \in \mathcal{M}_{k_{1},n_{1}}(\mathbb{F}_{q^m})$. 

\item A matrix $G_{2} \ra \mathcal{M}_{k_{2},n_{2}}(\mathbb{F}_{q^m})$.
\item Construct the matrix $G = G_{1} \otimes G_{2}$, which forms a generator matrix of Gabidulin-Kronecker product code $\mathcal{C}$. 
\item Let $H_{2} \in \mathcal{M}_{n_{2}-k_{2},n_{2}}(\mathbb{F}_{q^m})$ be a parity-check matrix of $\mathcal{C}_{2}$. 

\item A vector $\bm{a} \ra \mathbb{F}_{q^{m}}^{n}$, which forms a $k$-partial circulant matrix $X = Cir_{k}(\bm{a})\in \mathcal{M}_{k,n}(\mathbb{F}_{q^m})$ satisfying $Colr_{q}(X) = t_{1}$ and $0 < t_{1} \leq \lfloor \frac{n_{2}-k_{2}}{2} \rfloor$. 
\item A vector $\bm{b} \ra \mathbb{F}_{q^{m}}^{n}$, which forms an invertible circulant matrix $P = Cir_{n}(\bm{b})\in \mathcal{M}_{n,n}(\mathbb{F}_{q^m})$ with all entries in $\mathcal{V}$.
\item If the matrix composed of the first $k$ columns of matrix $(G+X)P^{-1}$ is invertible, then there exists an invertible circulant matrix $S = Cir_{k}(\bm{c}) \in \mathcal{M}_{k,k}(\mathbb{F}_{q^m})$ generated by a vector $\bm{c} \in \mathbb{F}_{q^m}^{k}$ such that $S(G+X)P^{-1}$ be of the systematic form. Otherwise, repeat the above process. \\
Compute $G_{pub} = S(G+X)P^{-1}$ and output the key pairs $pk=G_{pub}$ and $sk=(\bm{g_{1}},H_{2},\bm{b},\bm{c})$.
\end{enumerate}
\noindent $\Encrypt(\bm{m},pk)$. Let $\bm{m} = (\bm{m}_{1},\ldots,\bm{m}_{k_{1}})$ $\in \mathbb{F}_{q^m}^{k}$ be plaintext message where $\bm{m}_{i} \in \mathbb{F}_{q^m}^{k_{2}}$ for $1 \leq i \leq k_{1}$. Randomly choose a vector $\bm{e} = (\bm{e}_{1},\ldots,\bm{e}_{n_{1}}) \in \mathbb{F}_{q^m}^{n}$ such that $rk_{q}(\bm{e}) = t = \lfloor \frac{n_{2}-k_{2}-2t_{1}}{2\lambda} \rfloor$. Compute $\bm{c} = \bm{m}G_{pub} + \bm{e} = \bm{m}S(G+X)P^{-1}+\bm{e}$ and output $\bm{c}$.

\noindent $\Decrypt(\bm{c},sk)$. Compute $\bm{c}^{'} = \bm{c}P$ and perform $\mathfrak{D}_{\mathcal{C}} (\cdot)$, the Gabidulin-Kronecker product decoding algorithm described in \cite[Section 3.2]{SZZF2024} on $\bm{c}P$ and output $\bm{m} = \mathfrak{D}_{\mathcal{C}} (\bm{c}')$.

\noindent $\textbf{Correctness:}$ The correctness of our encryption scheme relies on
the decoding capability of the Gabidulin-Kronecker product code $\mathcal{C}$. To be specific, suppose that $\mathcal{C}$ decodes $\bm{c}^{'}$ correctly, we have Decrypt$(sk,$ Encrypt$(\bm{m}, pk)) = \bm{m}$.

 The suggested parameters for GabKron Cryptosystem can be shown in Table \ref{tab:4}.
\begin{table}[htbp]
\caption{Parameters and public key size (in bytes).}
\small
\label{tab:4}   
\begin{tabular}{lllllllllllll}
\hline\noalign{\smallskip}
Scheme  & $n_{1}$ & $k_{1}$ & $n_{2}$ & $k_{2}$ & $q$ & $m$ & $n$ & $k$ & $t$ & $\lambda$  & $pk$ size & Security \\
\noalign{\smallskip}\hline\noalign{\smallskip}
GabKron-128  & 2 & 2 & 24 & 12 & 2 & 48 & 48 & 24 & 12 & 3  & 288 & 128\\
GabKron-192  & 2 & 2 & 38 & 19 & 2 & 76 & 76 & 38 & 16 & 3  & 722 & 192\\
GabKron-256  & 2 & 2 & 52 & 26 & 2 & 104 & 104 & 52 & 24 & 3  & 1352 & 256\\
\noalign{\smallskip}\hline
\end{tabular}
\end{table}

\section{Analysis on the Gabidulin-Kronecker Product Code} \label{modify-gab-kron}

\noindent In this section, we recall the design of Gabidulin-Kronecker product code defined in \cite{SZZF2024}, and further analyse its structure. We show that Gabidulin-Kronecker product code is a subcode of Gabidulin-block code.

\subsection{Encoding and Decoding}

\begin{dfn}[{\cite{SZZF2024}}]
  For $s = 1$ and $2$, let $\mathcal{C}_i$ be an $[n_s,k_s,d_s]$ Gabidulin code with generator matrix $G_s = \left[ g_{s,ij} \right]_{1 \leq i \leq k_s, 1 \leq j \leq n_s}$ with error correcting capability $t_s = \left\lfloor \frac{n_s - k_s}{2} \right\rfloor$. Then an $[n,k]$-Gabidulin-Kronecker product code $\mathcal{C}$ is generated by 
  \[ G = G_1 \otimes G_2 = \left[ \begin{array}{ccc}
  g_{1,11} G_2 & \ldots & g_{1,1n_1} G_2 \\
  \vdots & \ddots & \vdots \\
  g_{1,k_11} G_2 & \ldots & g_{1,k_1n_1} G_2 \\
\end{array}   \right]  \]
where $n=n_1n_2$ and $k=k_1k_2$.
\end{dfn}

Sun et al. \cite{SZZF2024} introduced a decoding algorithm for the Gabidulin-Kronecker product code by performing block decoding on the received corrupted vector $\bm{y} = \bm{c} + \bm{e}$ where $\bm{c}$ is a codeword of $G$ and $\bm{e}$ is an error satisfying $\rk (\bm{e}) \leq t_2$. Note that since $G$ is a linear code, there exists $\bm{m} = (\bm{m}_1, \ldots, \bm{m}_{k_1}) \in \F_{q^m}^k$ where $\bm{m}_i \in \F_{q^m}^{k_2}$ for $1 \leq i \leq k_1$ such that $\bm{c} = \bm{m}G$. We can divide $\bm{c}$, $\bm{e}$ and $\bm{y}$ into $n_1$ equal blocks of length $n_2$, i.e. 
\begin{align*}
  \bm{y} = (\bm{y}_1,\ldots,\bm{y}_{n_1}) &= (\bm{m}_1,\ldots,\bm{m}_{k_1}) \left[ \begin{array}{ccc}
  g_{1,11} G_2 & \ldots & g_{1,1n_1} G_2 \\
  \vdots & \ddots & \vdots \\
  g_{1,k_11} G_2 & \ldots & g_{1,k_1n_1} G_2 \\
\end{array}   \right]  + (\bm{e}_1,\ldots,\bm{e}_{n_1}) \\
  &= \left( \sum_{i=1}^{k_1} \bm{m}_i g_{1,i1} G_2 + \bm{e}_1, \ldots,  \sum_{i=1}^{k_1} \bm{m}_i g_{1,i n_1} G_2 + \bm{e}_{n_1} \right).
\end{align*}
Let $\bm{c}_j = \displaystyle \sum_{i=1}^{k_1} \bm{m}_i g_{1,ij} G_2$ for $1 \leq j \leq n_1$. By performing decoding of $\mathcal{C}_2$ on each $\bm{y}_j = \bm{c}_j + \bm{e}_j$, we can recover $\bm{c}_j$ and $\bm{e}_j$. Let $I = \{ j_1,\ldots,j_{k_1} \} \subseteq \{ 1,2,\ldots,n_1\}$ be any information set of $G_1$. If $\rk(\bm{e}_i) \leq t_2$ for $i \in I$, we obtain the following system
\begin{align*}
  \bm{c}_{j_1} &= \sum_{i=1}^{k_1} \bm{m}_i g_{1,ij_1} G_2 \\
  &\vdots \\
  \bm{c}_{j_{k_1}} &= \sum_{i=1}^{k_1} \bm{m}_i g_{1,ij_{k_1}} G_2. \\
\end{align*}
If the above linear system is of full rank, i.e. rank $=k=k_1k_2$, then the vector $\bm{m}$ can be calculated.

\subsection{Subcode of Gabidulin-block code}

We now prove some other results for the Gabidulin-Kronecker product code. The following lemma demonstrates some necessary conditions for a Gabidulin-Kronecker product code.

\medskip

\begin{lem} \label{suffcond}
  Suppose that $G = G_1 \otimes G_2$ is a generator matrix for an $[n,k]$-Gabidulin-Kronecker product code defined as above. Then the matrix $\bar{G}_1 = \left[ \begin{array}{ccc}
  g_{1,11} I_{k_2} & \ldots & g_{1,1n_1} I_{k_2} \\
  \vdots & \ddots & \vdots \\
  g_{1,k_11} I_{k_2} & \ldots & g_{1,k_1n_1} I_{k_2} \\
\end{array}   \right] \in \F_{q^m}^{k_1k_2 \times n_1 k_2}$ has $rank(\bar{G}_1) = k$.
\end{lem}
\begin{proof}
  The matrix $G = G_1 \otimes G_2$ can be written in the following form:
\begin{align}
   G &= G_1 \otimes G_2 = \left[ \begin{array}{ccc}
  g_{1,11} G_2 & \ldots & g_{1,1n_1} G_2 \\
  \vdots & \ddots & \vdots \\
  g_{1,k_11} G_2 & \ldots & g_{1,k_1n_1} G_2 \\
\end{array}   \right] \notag \\
	&= \left[ \begin{array}{ccc}
  g_{1,11} I_{k_2} & \ldots & g_{1,1n_1} I_{k_2} \\
  \vdots & \ddots & \vdots \\
  g_{1,k_11} I_{k_2} & \ldots & g_{1,k_1n_1} I_{k_2} \\
\end{array}   \right] \left[ \begin{array}{ccc}
	G_2 & & \bm{0} \\
	& \ddots & \\
	\bm{0} & & G_2 \\
	\end{array} \right] := \bar{G}_1 \bar{G}_2 \label{defG1G2}
\end{align}
where $\bar{G}_1 \in \F_{q^m}^{k_1k_2 \times n_1 k_2}$ and $\bar{G}_2 \in \F_{q^m}^{n_1 k_2 \times n_1 n_2}$.

Since $G$ is an $[n,k]$-linear code, this implies that $rank(G) = k$. Recall that $rank(G) = \min \{rank(\bar{G}_1),rank(\bar{G}_2)\}$. As $\bar{G}_2$ is of diagonal form and $rank(G_2) = k_2$, then $rank(\bar{G}_2) = n_1 \times k_2$. On the other hand, $rank(\bar{G}_1) \leq k = k_1 k_2 \leq n_1 k_2 = rank(\bar{G}_2)$, we have $rank(\bar{G}_1) = k$.
\end{proof}

\begin{dfn}
  Let $n=n' n_2$, $k=n' k_2$ and $\rk(\bm{g}_i) = n_2$ where $\bm{g}_i \in \F_{q^m}^{n_2}$ for $1 \leq i \leq n'$. An $[n,k]$-Gabidulin-block code with components $\bm{g}_1,\ldots,\bm{g}_{n'}$ has generator matrix in the form of
  \[ G =  \left[ \begin{array}{ccc}
	\Gab(\bm{g}_1) & & \bm{0} \\
	& \ddots & \\
	\bm{0} & & \Gab(\bm{g}_{n'}) \\
	\end{array} \right]  \]
	where each $\Gab(\bm{g}_i)$ is an $[n_2,k_2]$-Gabidulin code.
\end{dfn}

\medskip

\begin{cor}
An $[n,k]$-Gabidulin-Kronecker product code $\mathcal{C}$ is a subcode of $[n,k]$-Gabidulin-block code $\mathcal{C}'$ with components $\bm{g}_2,\ldots,\bm{g}_{2}$.
\end{cor}
\begin{proof}
Let $\bm{c}$ be a codeword of $\mathcal{C}$, then there exists $\bm{m} \in \F_{q^m}^k$ such that $\bm{c} = \bm{m}G = \bm{m} \bar{G}_1 \bar{G}_2$. Let $\bm{\mu} = \bm{m} \bar{G}_1$, then $\bm{c} = \bm{\mu} \bar{G}_2$. Since $\bar{G}_2$ is a generator matrix for $\mathcal{C}'$, we have $\bm{c} \in \mathcal{C}'$.
\end{proof}

Noticeably that the Moore structure of $G_1$ in fact does not play any role in the decoding phase, and hence the absence of decoding algorithm for $\mathcal{C}_1$ in the decoding of $\mathcal{C}$. As such, we could consider $G_1$ to be any random code, as long as the sufficient condition in Lemma \ref{suffcond} is satisfied.  

\section{Analysis on the Design of GabKron Cryptosystem}

\subsection{On the choice of $X$ } \label{Xchoice}

In the $\Keygen$ of GabKron cryptosystem, Sun et al. did not specify on how to choose the vector $\bm{a}$ and matrix $X$. We show that such choice of $\bm{a}=(a_0,\ldots,a_{n-1}) \in \F_{q^m}^n$ such that $X = \circu{k}{\bm{a}}$, $\colrk_q (X) = t_1$ and $0 < t_1 \leq \left\lfloor \frac{n_2-k_2}{2} \right\rfloor$ is possible for certain parameters.

\medskip

For simplicity, we consider $n= n' t_1$. To construct such a matrix $X$ such that $\colrk_q (X) = t_1$, we can determine a matrix $W =[w_{ij}] \in \F_q^{n \times n}$ such that $X = \left[ Y \mid \bm{0}_{k \times (n-t_1)} \right]W$ where $Y \in \F_{q^m}^{k \times t_1}$. 

\medskip

Consider $\displaystyle W = \left[ \begin{array}{c}
  W_T \\
  \hline
  W_B
\end{array} \right]$ and $W_T = \left[ W_1 | \ldots | W_{n'} \right]$ where $W_i \in \F_{q^m}^{t_1 \times t_1}$ for $1 \leq i \leq n'$.
We randomly choose $\bm{y}_1 \ra \F_{q^m}^{t_1}$ and let $\bm{y}_1 = (y_{11},\ldots,y_{1t_1})$ be the first row of $Y$, then the first row of $X$ is
\begin{align}
  (a_0,a_{n-1},\ldots,a_1) &= \bm{y}_1 W_T = (y_{11},\ldots,y_{1t_1}) \left[ W_1 | \ldots | W_{n'} \right], \label{Xrow1}
\end{align}
thus obtaining $\bm{a}$.

\medskip

To compute $\bm{y}_2$, the second row of $Y$, we consider the second row of $X$:
\begin{align}
  (a_1, a_0,\ldots,a_2) &= \bm{y}_2 W_T = (y_{21},\ldots,y_{2t_1}) \left[ W_1 | \ldots | W_{n'} \right]. \label{Xrow2}
\end{align}
Note that $(a_1,a_0,\ldots,a_2)$ can also be computed by shifting 1 position of the equation (\ref{Xrow1}) to the right, i.e.
\begin{align}
  W'_T &= \left[ \begin{array}{c|c|c|c|ccc}
  w_{1,n't_1} & & & & w_{1,(n'-1)t_1+1} & \ldots & w_{1,n't_1-1} \\
  \vdots & W_1 & \ldots & W_6 &  \\
  w_{t_1,n't_1} & & & & w_{t_1,(n'-1)t_1+1} & \ldots & w_{t_1,n't_1-1} \\
\end{array}    \right] =: \left[ W'_1 | \ldots | W'_{n'} \right] \notag \\
  (a_1, a_0,\ldots,a_2) &= \bm{y}_1 W'_T = (y_{11},\ldots,y_{1t_1}) \left[ W'_1 | \ldots | W'_{n'} \right] \label{Xrow2-1}
\end{align}
Equating (\ref{Xrow2}) and (\ref{Xrow2-1}), we have
\begin{align*}
  \bm{y}_2 W_T &= \bm{y}_1 W'_T \\
  \bm{y}_2 \left[ W_1 | \ldots | W_{n'} \right] &= \bm{y}_1 \left[ W'_1 | \ldots | W'_{n'} \right].
\end{align*}
Assume that each $W_i$ and $W'_i$ are invertible, then we can compute $\bm{y}_2 = \bm{y}_1 W'_i W_i^{-1}$. This implies that the matrix $W$ needs to satisfy the conditions $W'_i W_i^{-1} = W'_j W_j^{-1}$ for $1 \leq i,j \leq n'$.

\medskip

We can repeat the steps above to compute $\bm{y}_j$ from $\bm{y}_{j-1}$ for $3 \leq j \leq k$.

\medskip

\begin{rem}
  The steps illustrated above may be generalized to other parameters for GabKron cryptosystem. However, as observed from the conditions $W'_i W_i^{-1} = W'_j W_j^{-1}$, the choices for such construction of $X$ may not be many.
\end{rem}

\subsection{On the Left Scrambler Matrix $S$} \label{matS}

Sun et al. claimed that there exists an invertible circulant matrix $S \in \F_{q^m}^{k \times k}$ such that $\Gpub = S[G+X]P^{-1}$ is of systematic form.  We will show that this statement is generally not true.

\medskip

\begin{prop} \label{Gpubsystem}
  Let $M = [M_1 | M_2] \in \F_{q^m}^{k \times n}$ such that $M_1 \in \F_{q^m}^{k \times k}$ is invertible. There exists an invertible circulant matrix $S$ such that $SM$ is of systematic form if and only if $M_1$ is a circulant matrix.
\end{prop}
\begin{proof}
  Suppose that $S$ is circulant and $SM$ is of systematic form, i.e. $SM = \left[ I_k | N \right]$ where $N=SM_2 \in \F_{q^m}^{k \times (n-k)}$. This implies that $M_1 = S^{-1}$ is circulant, since the inverse of a circulant matrix is circulant.
  
  \medskip
  
  If $M_1$ is circulant invertible matrix, then there exists a circulant matrix $S$ such that $SM_1 = I_k$, thus giving $SM$ in its systematic form $SM=[I_k | SM_2 ]$.
\end{proof}

\medskip

By the specifications in $\Keygen$, it is not guaranteed that there exist $k$ linearly independent columns of $[G+X]P^{-1}$ such that they are of circulant structure. Therefore, by Proposition \ref{Gpubsystem}, such left scrambler matrix $S$ with circulant invertible structure does not exist.

\subsection{On the Key Sizes}

\noindent \textbf{Public Key Size}: The public key $\pk$ of GabKron, $\Gpub$ is of systematic form. Note that there exists invertible matrix $S$ such that $\Gpub = S[G+X]P^{-1} = [I_k \mid N ]$ where $N \in \F_{q^m}^{k \times (n-k)}$, although $S$ may not be of circulant structure. This implies that the public key size should be $\size_\pk = \frac{mk(n-k)}{8} = \frac{mk_1k_2(n_1n_2 - k_1k_2)}{8}$ bytes.

\medskip 

\noindent \textbf{Secret Key Size}: The secret key $\sk$ of GabKron is $(\bm{g}_1, H_2, \bm{b}, \bm{c})$. Sun et al. proposed the used of a normal element $\alpha$ to generate the vector $\bm{g}_1$, which requires $m$ bits in storing $\alpha$. For $H_2$, there exists a vector $\bm{h}_2 \in \F_{q^m}^{n_2}$ such that $H$ is a Moore matrix generated by $\bm{h}_2$. Therefore it requires $mn_2$ bits to store $\bm{h}_2$. For $\bm{b}$, since $\rk(\bm{b}) = \lambda$, we can write $\bm{b} = (B_1,\ldots,B_\lambda)$ where $B_1,\ldots,B_\lambda \in \F_{q^m}$ and $M_B \in \F_q^{\lambda \times n}$, which requires $\lambda (m+n)$ bits for storage. Finally, the vector $\bm{c}$ is used to store the secret key $S$. As argued in Section \ref{matS}, $S$ should not be a circulant matrix, thus requiring $k \times k \times m$ bits for storage. Therefore, the secret key size should be
\[ \size_\sk = \frac{m(1+n_2 + \lambda + k^2 ) + \lambda n}{8} \text{ bytes}. \]

\subsection{On the Parameter Sets}

The parameters for GabKron seems very optimal as compared to other code-based encryption schemes at similar security level. For instance, GabKron requires only $288$ bytes for public key at 128-bit security level, much lower than its rank metric public-encryption scheme counterparts. However, we would like to show that parameters provided are in fact impractical.

\medskip

Recall that $\rk(\bm{e}) = t = \displaystyle \left\lfloor \frac{n_2-k_2-2t_1}{2\lambda} \right\rfloor$. We first note that the value of $t_1$ is not given in the parameter sets. Since $t_1 \geq 0$, we know that $t = \displaystyle \left\lfloor \frac{n_2-k_2-2t_1}{2\lambda} \right\rfloor \leq  \left\lfloor \frac{n_2-k_2}{2\lambda} \right\rfloor$. By considering the value of $n_2$, $k_2$, $\lambda$ and $t$ given in GabKron-128, GabKron-192 and GabKron-256, we notice that the inequality for $t$ is not satisfied:

\begin{table}[h!]
  \centering
   \begin{tabular}{|c|ccc|c|c|}
    \hline
      Scheme & $n_2$ & $k_2$ & $\lambda$ & $t$ & $\left\lfloor \frac{n_2-k_2}{2\lambda} \right\rfloor$ \\
      \hline
      GabKron-128 & 24 & 12 & 3 & 12 & 2 \\
      GabKron-192 & 38 & 19 & 3 & 16 & 3 \\
      GabKron-256 & 52 & 26 & 3 & 24 & 4 \\
    \hline
  \end{tabular}
\end{table}

\noindent  In other words, the three parameter sets provided in \cite{SZZF2024} are not practical and not realistic to achieve the desired security level.

\subsection{Repaired Parameter Sets}

By taking all the considerations above, we now propose some repaired parameter sets for GabKron cryptosystem.

\medskip

As discussed in Section \ref{Xchoice}, we choose $n = 7t_1$ and $k=3t_1$. Furthermore, since the decoding of $\mathcal{C}_1$ is not required in the decryption process, we choose $\mathcal{C}_1$ to be a random code with $k_1 = n_1$.
\begin{table}[h!]
  \centering
   \begin{tabular}{|c|cccc|ccccccc|cc|}
    \hline
      Scheme & $n_1$ & $k_1$ & $n_2$ & $k_2$ & $q$ & $m$ & $n$ & $k$ & $t$ & $t_1$ & $\lambda$ & $\size_\pk$ & {\sf Sec} \\
      \hline
      Rep-GabKron-128 & 2 & 2 & 105 & 35 & 2 & 211 & 210 & 70 & 9 & 7 & 3 & 258.5K & 128 \\
      Rep-GabKron-192 & 2 & 2 & 150 & 50 & 2 & 307 & 300 & 100 & 13 & 10 & 3 & 767.5K & 192  \\
      Rep-GabKron-256 & 2 & 2 & 165 & 55 & 2 & 331 & 330 & 110 & 14 & 11 & 3 & 1001.3K & 256  \\
    \hline
  \end{tabular}
\end{table}

\section{Improved GabKron Cryptosystem with Partial-Circulant-Block Structure}

In this section, we propose an improvement to utilize the design of GabKron cryptosystem by considering partial circulant block structure for the public key. Furthermore, we propose a new approach in choosing the right scrambler matrix $P$ so that the rank weight of error chosen can be increased as compared to the original GabKron.

\subsection{General Idea for Improved GabKron Cryptosystem}

Recall in the decryption algorithm:
\begin{align}
  \bm{c} = \bm{m}\Gpub + \bm{e} &= \bm{m}S [G+X] P^{-1} + \bm{e} \notag \\
  \bm{c}' = \bm{c}P &= \bm{m}S [G+X] + \bm{e}P  \notag  \\
  (\bm{c}'_1,\ldots,\bm{c}'_{n1}) &= \bm{m}S \left[ \begin{array}{ccc}
  g_{1,11} G_2 & \ldots & g_{1,1n_1} G_2 \\
  \vdots & \ddots & \vdots \\
  g_{1,k_11} G_2 & \ldots & g_{1,k_1n_1} G_2 \\
\end{array}   \right] + \bm{m}S \left[ \begin{array}{ccc}
  X_{11} & \ldots & X_{1n_1} \\
  \vdots & \ddots & \vdots \\
  X_{k_11} & \ldots & X_{k_1n_1}  \\
\end{array}   \right] \notag \\
  &\qquad + \bm{e} \left[ \begin{array}{ccc}
  P_{11} & \ldots & P_{1n_1} \\
  \vdots & \ddots & \vdots \\
  P_{n_11} & \ldots & P_{n_1n_1}  \\
\end{array}   \right] \notag \\
\Rightarrow \quad \bm{c}'_i &= \bm{m}S (g_{1,1i} + \ldots +  g_{1,k_1i}) G_2 + \bm{m}S \left[ \begin{array}{c}
  X_{1i} \\
  \vdots \\
  X_{k_1i}  \\
\end{array}  \right] + \bm{e} \left[ \begin{array}{c}
  P_{1i} \\
  \vdots \\
  P_{n_1i}  \\
\end{array}   \right] \notag
\end{align}
where $X_{ij} \in \F_{q^m}^{k_2 \times n_2}$ and $P_{ij} \in \mathcal{V}^{n_2 \times n_2}$. Since decryption is performed by block decoding, it suffices for us to focus on each block $\bm{c}'_i$ and improve the cryptosystem accordingly.

\medskip

Observe that the matrix $G_1$ acts as left scrambler for the matrix, therefore we can choose to omit the matrix $S$ as part of the left scrambler matrix. 

\medskip

Denote $X_{C_i}$ and $P_{C_i}$ as the $i$th column block of $X$ and $P$ respectively, i.e.
\[ X_{C_i} := \left[ \begin{array}{c}
  X_{1i} \\
  \vdots \\
  X_{k_1i}  \\
\end{array}  \right] \quad \text{and} \quad P_{C_i} := \left[ \begin{array}{c}
  P_{1i} \\
  \vdots \\
  P_{n_1i}  \\
\end{array}   \right]. \]
For block-decoding with respect to $\mathcal{C}_2$ to work, we need to ensure that 
\[ \rk( \bm{m}S X_{C_i})\leq t_1, \quad \rk(\bm{e}P_{C_i})  \leq t \lambda \quad \Rightarrow \quad  \rk( \bm{m}S X_{C_i}) + \rk(\bm{e}P_{C_i})  \leq t_2. \]

\medskip

\noindent Our idea of improvement for GabKron is as follows:

\medskip

\begin{enumerate}
\item Constructing $X_{C_i}$ such that each of the blocks has $\colrk_q(X_{C_i}) = t_1$. This will enable block-decoding and result in the whole matrix $X$ has $\colrk_q (X) \leq n_1 t_1$ (instead of $\colrk_q(X) \leq t_1$);
\item Constructing $P_{C_i}$ such that $\rk( \bm{e}P_{C_i}) \leq \lambda't$ where $\lambda' \leq \lambda$ (instead of $\rk( \bm{e}P_{C_i}) \leq \lambda t$) and $P \in \mathcal{V}^{n \times n}$;
\item Utilize the structure of partial-circulant-block matrices to optimize the key sizes.
\end{enumerate}

\medskip

\noindent The details of the constructions will be specified in sections later.

\subsection{Modification on $G$}

We now give the definition for circulant-block and partial-circulant-block matrices.

\medskip

\begin{dfn}
Let $n=n_1n_2$ and $k = k_1k_2$ where $1 \leq k_1 \leq n_1$ and $1 \leq k_2 \leq n_2$. A square block matrix $A  \in \F_{q^m}^{n \times n}$ be a square matrix
\[ A = \left[ \begin{array}{ccc}
    A_{11} & \ldots & A_{1n_1}  \\
    \vdots & \ddots & \vdots \\
    A_{n_11} & \ldots & A_{n_1n_1}  \\
\end{array} \right], \qquad A_{ij} \in \F_{q^m}^{n_2 \times n_2} \text{ for } \leq i, j \leq n_1 \]
is called a circulant-block matrix if each block $A_{ij}$ is circulant. We say that a matrix $B \in \F_{q^m}^{k \times n} $ is partial-circulant-block matrix if 
\[ B= \left[ \begin{array}{ccc}
    B_{11} & \ldots & B_{1n_1}  \\
    \vdots & \ddots & \vdots \\
    B_{k_11} & \ldots & B_{k_1n_1}  \\
\end{array} \right], \qquad B_{ij} \in \F_{q^m}^{k_2 \times n_2} \text{ for } \leq i \leq k_1, 1 \leq j \leq n_1,  \]
and each $B_{ij}$ is a $k_2$-partial circulant matrix.
\end{dfn}

\begin{lem}\cite[Lemma 1]{Rjasa94} \label{circublockinverse}
    The set of all invertible circulant-block matrices of the same dimension $n=n_1n_2$ is a group with respect to matrix multiplication.
\end{lem}

\begin{lem} \cite[Proposition 1]{LT19} \label{productpcircu}
   Let $P$ be a $k$-partial circulant matrix and $Q$ be a circulant matrix, then $PQ$ is a $k$-partial circulant matrix.
\end{lem}

\begin{lem} \label{circublockproduct}
    Let $n=n_1n_2$ and $k = k_1k_2$ where $1 \leq k_1 \leq n_1$ and $1 \leq k_2 \leq n_2$. Suppose that $A$ is a circulant-block matrix where each block $A_{ij} \in \F_{q^m}^{n_2 \times n_2}$ is a circulant matrix where $1 \leq i,j \leq n_1$. Let $B$ be a partial-circulant-block matrix where each block $B_{ij}$ is a $k_2$-partial circulant matrix where $1 \leq i \leq k_1$ and $1 \leq j \leq n_1$. Then the matrix $Q=BA$ is a partial-circulant-block matrix where each $Q_{ij}$ is a $k_2$-partial circulant matrix.
\end{lem}
\begin{proof}
    Consider the matrix $Q= BA$:
    \begin{align*}
        Q = \left[ \begin{array}{ccc}
    Q_{11} & \ldots & Q_{1n_1}  \\
    \vdots & \ddots & \vdots \\
    Q_{k_11} & \ldots & Q_{k_1n_1}  \\
\end{array} \right] &= \left[ \begin{array}{ccc}
    B_{11} & \ldots & B_{1n_1}  \\
    \vdots & \ddots & \vdots \\
    B_{k_11} & \ldots & B_{k_1n_1}  \\
\end{array} \right] \left[ \begin{array}{ccc}
    A_{11} & \ldots & A_{1n_1}  \\
    \vdots & \ddots & \vdots \\
    A_{n_11} & \ldots & A_{n_1n_1}  \\
\end{array} \right] = BA \\
    Q_{ij} &= B_{i1} A_{1j} + B_{i2} A_{2j} + \ldots + B_{in_1} A_{n_1j} = \sum_{k_0 = 1}^{n_1} B_{i k_0} A_{k_0 j}.
    \end{align*}  
    It suffices to show that $Q_{ij}$ is a $k_2$-partial circulant matrix. By Lemma \ref{productpcircu}, each $B_{i k_0} A_{k_0 j}$ is a $k_2$-partial circulant matrix, and thus the sum of partial circulant matrices $Q_{ij}$ is a partial circulant matrix. 
\end{proof}

\medskip

Recall the structure of the matrix $G$ in (\ref{defG1G2}):

\medskip

\[
   G = G_1 \otimes G_2 = \left[ \begin{array}{ccc}
  g_{1,11} G_2 & \ldots & g_{1,1n_1} G_2 \\
  \vdots & \ddots & \vdots \\
  g_{1,k_11} G_2 & \ldots & g_{1,k_1n_1} G_2 \\
\end{array}   \right]. \]
If $G_2$ is of partial circulant structure, then the matrix $G$ is a partial-circulant-block matrix. Note that the discussion in Section \ref{modify-gab-kron} allows us to choose $G_1$ to be of random structure with full rank $k_1$.

\subsection{Construction of $X$} \label{ConstructX}

Since $G$ is of partial-circulant-block form, to optimize the public key size, we consider $X$ to be of partial-circulant-block form as well, i.e.
\medskip
\[  X = \left[ \begin{array}{ccc}
  X_{11} & \ldots & X_{1n_1} \\
  \vdots & \ddots & \vdots \\
  X_{k_11} & \ldots & X_{k_1n_1}  \\
\end{array}   \right]  \]
where each $X_{ij}$ is a partial circulant matrix with $\colrk_q (X_{ij}) = t_1$. Let $I$ be an information set of $G_1$. For $j \notin I$, the matrix $X_{ij}$'s can be chosen to be random $k_2$-partial circulant matrix.

For $j \in I$, the $X_{ij}$'s can be generated by the steps in Section \ref{Xchoice}. Let $\bm{\mu} = \bm{m}S = (\bm{\mu_1},\ldots,\bm{\mu_{k_1}})$ where $\bm{\mu_{i}} \in \F_{q^m}^{k_2}$. Since each $\colrk_q(X_{ij}) = t_1$, there exist $Y_{ij} \in \F_{q^m}^{k_2 \times t_1}$ $W_{ij} \in \F_q^{n_2 \times n_2}$ such that $X_{ij} = \left[ Y_{ij} \mid \bm{0}_{k_2 \times (n_2-t_1)} \right] W_{ij}$. We have
\begin{align*}
    \bm{m}S X_{C_j} &= (\bm{\mu_1},\ldots,\bm{\mu_{k_1}}) \left[ \begin{array}{c}
  X_{1j} \\
  \vdots \\
  X_{k_1j}  \\
\end{array}  \right] = \sum_{i=1}^{k_1} \bm{\mu_i} X_{ij} = \sum_{i=1}^{k_1} \bm{\mu_i} \left[ Y_{ij} \mid \bm{0}_{k_2 \times (n_2-t_1)} \right] W_{ij},
\end{align*}
which may gives us $t_1 \leq \rk(\bm{m}SX_{C_j}) $.

\medskip

\noindent To ensure $ \rk(\bm{m}SX_{C_j}) \leq t_1 $, we can consider $W'_j  = W_{i_1 j} = W_{i_2j}$ where $1 \leq i_1, i_2 \leq k_1$, then
\[ \bm{m}S X_{C_j}  =  \sum_{i=1}^{k_1} \bm{\mu_i} \left[ Y_{ij} \mid \bm{0}_{k_2 \times (n_2-t_1)} \right] W_{ij} = \left( \sum_{i=1}^{k_1} \bm{\mu_i} Y_{ij} \mid \bm{0}_{k_2 \times (n_2-t_1)} \right) W'_j \]
which gives us $\rk( \bm{m}S X_{C_j} ) \leq t_1$.

\subsection{Construction of $P$} \label{ConstructP}

In the original design of GabKron cryptosystem, $P$ is chosen over $\mathcal{V}$, an $\F_q$-subspace of $\F_{q^m}$ with dimension $\lambda$. Let $\{ \gamma_1,\ldots,\gamma_\lambda \}$ be a basis of space $\mathcal{V}$. Let $1 \leq \lambda' \leq \lambda$, $I = \{ j_1,\ldots,j_{k_1} \} \subseteq \{ 1,\ldots,n_1\}$ be an information set of $G_1$.

\medskip

For each $i \in I$, randomly select $\Gamma_{i,1},\ldots,\Gamma_{i,\lambda'} \ra \{ \gamma_1,\ldots,\gamma_\lambda\}$, and form a subspace $\mathcal{U}_i = \langle \Gamma_{i,1} ,\ldots,\Gamma_{i,\lambda'} \rangle \subseteq \mathcal{V}$. Then we can form a matrix $P_{C_i}$, whose components are in $\mathcal{U}_i$. This will ensure that $\rk_q (\bm{e} P_{C_i}) \leq \rk_q (\bm{e}) \times \lambda' = \lambda' t \leq \lambda t$. For $i \notin I$, we can construct a matrix $P_{C_i}$, whose components are taken from $\mathcal{V}$.

\medskip

To optimize the public key size, we choose the matrices $P_{ij}$ to be $n_2 \times n_2$ circulant matrices. Finally, check the singularity of matrix $P$ and ensure that it is invertible.

\subsection{Improved GabKron Cryptosystem}

We now give the specification for our proposed GabKron cryptosystem.

\medskip

\noindent $\Setup(\delta)$. For a given security level $\delta$, choose $q$, a power of prime, and positive integers $m,n,k,n_1,n_2,k_2,k_2,t,t_1,t_2,\lambda,\lambda'$ satisfying $n_2 = m$, $n=n_1n_2$, $k=k_1k_2$, $1 \leq k_1 \leq n_1$, $1 \leq k_2 \leq n_2$, $2 \leq \lambda' \leq \lambda$, $t_2 = \left\lfloor \frac{n_2-k_2}{2} \right\rfloor$ and $\lambda' t + t_1 \leq t_2$. Output parameters ${\sf param} = (q,m,n,k,n_1,n_2,k_2,k_2,t,t_1,t_2,\lambda,\lambda')$.

\medskip

\noindent $\Keygen(q,m,n,k,n_1,n_2,k_2,k_2,t,t_1,t_2,\lambda,\lambda')$. Randomly generate the following:
\begin{enumerate}
    \item A matrix $G_1 \ra \F_{q^m}^{k_1 \times n_1}$.
    \item A normal element $\alpha \ra \F_{q^m}$. Let $\bm{g}_2 = (\alpha^{[n_2-1]}, \alpha^{[n_2-2]}, \ldots, \alpha)$. Form a generator matrix $G_2 = \circu{k_2}{\bm{g_2}}$.
    \item A partial-circulant-block matrix $X$ generated by the steps in Section \ref{Xchoice} and \ref{ConstructX}.
    \item A circulant-block matrix $P$ generated by the steps in Section \ref{ConstructP}.
 \end{enumerate}
Compute $\Gpub = [G+X]P^{-1}$ and output the key pairs $\pk = \Gpub$ and $\sk = ( \alpha, P, G_1 )$.

\medskip

\noindent $\Encrypt(\bm{m},\pk)$. Let $\bm{m} = (\bm{m}_1,\ldots,\bm{m}_{k_1}) \in \F_{q^m}^{k} $ be a plaintext message where $\bm{m}_j \in \F_{q^m}^{k_2}$ for $ 1\leq j \leq k_1$. Randomly choose a vector $\bm{e}= (\bm{e}_1,\ldots,\bm{e}_{n_1}) \in \F_{q^m}^n$ such that $\rk(\bm{e}) = t$ and $\bm{e}_i \in \F_{q^m}^{n_2}$ for $1 \leq i \leq n_1$. Compute $\bm{c} = \bm{m}\Gpub + \bm{e}$ and output $\bm{c}$.

\medskip

\noindent $\Decrypt(\bm{c},\sk)$. Compute $\bm{c}' = \bm{c}P$ and perform $\mathfrak{D}_{\mathcal{C}}(\cdot)$, the decoding algorithm of the Gabidulin-Kronecker product code $\mathcal{C}$ on $\bm{c}P$ and output $\bm{m} = \mathfrak{D}_G (\bm{c}')$.

\medskip

\noindent \textbf{Correctness}. The correctness of our encryption scheme lies on the decoding capability of the Gabidulin-Kronecker product code $\mathcal{C}$ and the choice for the right scrambler $P$. In particular,

\medskip

\[ \bm{c}' = \bm{c} P = \bm{m}[G+X] + \bm{e}P, \quad \bm{c}'_i = \bm{m}(g_{1,1i} + \ldots + g_{1,k_1i}) G_2 + \bm{m}X_{C_i} + \bm{e}P_{C_i} \]
where $1 \leq i \leq n_1$. For each $i \in I$ an information set of $G_1$,
\[ \rk(\bm{m} X_{C_i} + \bm{e}P_{C_i}) \leq t_1 + \lambda' t \leq t_2. \]
Therefore we can compute $\bm{m}$ by the decoding algorithm of $\mathcal{C}$.

\subsection{Parameters for Improved GabKron Cryptosystem}

Before we present our new parameters, we first compute $\size_\pk$, the public key size, and $\size_\sk$, the secret key size.

\medskip

\noindent \textbf{Public Key Size}: Since $P$ is a circulant-block matrix, the inverse of $P$, by Lemma \ref{circublockinverse}, $P^{-1}$ is also a circulant-block matrix. Note that both $G$ and $X$ are partial-circulant-block matrix where each block is a $k_2$-partial circulant matrix, thus $G+X$ is also a partial-circulant-block matrix. By Lemma \ref{circublockproduct}, $\Gpub = [G+X]P^{-1}$ is a partial-circulant-block matrix. Therefore

\medskip

\[ \size_\pk = \frac{1}{8} k_1 n_1 n_2 m \text{ bytes}. \]

\medskip

\noindent \textbf{Secret Key Size}: Recall that the secret key is $(\alpha, P , G_1)$, $P$ is a circulant-block matrix. Note that $P$ can be represented by $n_1^2$ vectors of length $n_2$ each, and these vectors has rank $\lambda$. Therefore

\medskip

\[ \size_\sk = \frac{1}{8} \left( m + n_1^2 \lambda (m+n_2) + k_1 n_1 m \right) \text{ bytes}.\]

\medskip

\begin{table}[h]
  \centering
  \label{impro}
  \small
   \begin{tabular}{|c|cccc|ccccccc|cc|}
    \hline
      Scheme & $n_1$ & $k_1$ & $n_2$ & $k_2$ & $q$ & $m$ & $n$ & $k$ & $t$ & $t_1$ & $\lambda$ & $size_{pk}$ & Sec \\
      \hline
      new-GabKron-128 & 2 & 2 & 90 & 18 & 2 & 90 & 180 & 36 & 12 & 6 & 3 & 4050 & 128 \\
      new-GabKron-192 & 2 & 2 & 120 & 32 & 2 & 120 & 240 & 64 & 14 & 8 & 3 & 7200 & 192  \\
      new-GabKron-256 & 2 & 2 & 128 & 40 & 2 & 128 & 256 & 80 & 14 & 8 & 3 & 8192 & 256  \\
    \hline
  \end{tabular}
\end{table}


\begin{thebibliography}{99}


\bibitem{LT19}
Lau T S C, Tan C H. New rank codes based encryption scheme using partial circulant matrices. Des Codes Cryptogr, 2019, 87: 2979-2999.


\bibitem{Rjasa94}
Rjasanow S. Effective Algorithms with Circulant-Block Matrices. Linear Algbra and Its Applications, 1994, 202: 55-69.



\bibitem{SZZF2024}
Sun Z, Zhuang J, Zhou Z, et al. A new McEliece-type cryptosystem using Gabidulin-Kronecker product codes. Theoretical Computer Science, 2024: 114480.


\end{thebibliography}
\end{document}